\documentclass[pra,twocolumn,nofootinbib,superscriptaddress,nobalancelastpage]{revtex4-1}
\usepackage{color,amsthm,amsmath,amsfonts,graphicx,bm,enumitem,bbm}

\def\Id{{\openone}}
\newcommand{\be}{\begin{equation}}
\newcommand{\ee}{\end{equation}}
\newcommand{\bea}{\begin{eqnarray}}
\newcommand{\eea}{\end{eqnarray}}
\newcommand{\E}{\mathcal{E}}
\newcommand{\1}{\mathbbm{1}}

\newtheorem{thm}{Theorem}
\newtheorem{prop}[thm]{Proposition}
\newtheorem{cor}[thm]{Corollary}
\newtheorem{lem}[thm]{Lemma}
\newtheorem{defn}[thm]{Definition}

\newcommand{\ket}[1]{\vert#1\rangle}

\newcommand{\tr}{\mathrm{tr}}
\newcommand{\mc}[1]{\mathcal{#1}}

\begin{document}

\title{Robustness in Projected Entangled Pair States}

\author{J.~Ignacio \surname{Cirac}}
\affiliation{Max-Planck-Institut f{\"{u}}r Quantenoptik,
Hans-Kopfermann-Str.\ 1, D-85748 Garching, Germany}

\author{Spyridon Michalakis}
\affiliation{Institute for Quantum Information and Matter, Caltech, 91125 Pasadena, CA, U.S.A.}

\author{David P\'erez-Garc\'\i a}
\affiliation{Dpto.\ An\'alisis Matem\'atico and IMI,
Universidad Complutense de Madrid, 28040 Madrid, Spain}

\author{Norbert Schuch}
\affiliation{Institute for Quantum Information, RWTH Aachen University,
D-52056 Aachen, Germany}

\begin{abstract}
We analyze a criterion which guarantees that the ground states
of certain many body systems are stable under perturbations. Specifically,
we consider PEPS, which are believed to provide an efficient description, based on local tensors, for the low energy physics arising
from local interactions.  In order to assess stability in the framework of
PEPS, one thus needs to understand how physically allowed perturbations of
the local tensor affect the properties of the global state.  In this
paper, we show that a restricted version of the Local Topological Quantum
Order (LTQO) condition~\cite{michalakis:local-tqo-ffree} provides a
checkable criterion which allows to assess the stability of local
properties of PEPS under physical perturbations.  We moreover show that
LTQO itself is stable under perturbations which preserve the spectral gap,
leading to nontrivial examples of PEPS which possess LTQO and are thus
stable under arbitrary perturbations.  \end{abstract}

\maketitle

\section{Introduction}

In studying model Hamiltonians for condensed matter systems, it is
essential to understand the conditions which guarantee that properties
behave nicely under small perturbations, as this allows to use the model
to predict the behavior of actual physical systems.  In the context of
zero-temperature physics, this amounts to understanding the conditions
under which certain physical properties of the ground state change
smoothly under perturbations to the Hamiltonian.  While this question is
very hard to answer in general, a proof of stability under arbitrary
perturbations has recently been given for frustration-free
Hamiltonians~\cite{michalakis:local-tqo-ffree}, based on two conditions  (LTQO and local gap),
following up on earlier work on commuting
Hamiltonians~\cite{bravyi:tqo-long,bravyi:local-tqo-simple}. However, the
LTQO condition, and especially the local gap condition, are very hard
to check in practice, and so far, no examples beyong commuting
Hamiltonians fulfilling these properties have been devised.

Projected Entangled Pair States (PEPS) provide a local description of
quantum many-body states based on their entanglement structure, and thus
in a natural way embody the physics of local gapped Hamiltonians. PEPS can
be used as a framework to understand the physics of many-body systems
based on the state (similar as e.g.\ the Laughlin
wavefunction), in particular since to any PEPS, a local \emph{parent
Hamiltonian} can be associated.  In the case of translational invariant
systems, the state is described by a single local tensor, and
understanding any property of the system can be mapped to studying a
corresponding property of this tensor.  In this way, PEPS have been very
successful in understanding otherwise intractable questions, such as the
characterization of topological order from local
symmetries~\cite{schuch:peps-sym}, the way in which global symmetries
emerge locally~\cite{perez-garcia:inj-peps-syms}, or the characterization
of quantum phases without and with symmetries in one
dimension~\cite{pollmann:1d-sym-protection-prb,chen:1d-phases-rg,schuch:mps-phases}
and beyond~\cite{chen:2d-spt-phases-peps-ghz}, just to name a few. To
assess how general these result are, it is therefore important to
identify the conditions under which PEPS are robust to perturbations.
Given the state-centered perspective of the PEPS framework, we are
particularly interested in those ``natural'' perturbations to the state
which both correspond to a perturbation of the local tensor, and at the
same time can be understood as arising from a perturbation of the parent
Hamiltonian.  Unfortunately, the powerful tools available to assess these
questions in one dimension cannot be applied to two-dimensional systems,
leaving the stability of PEPS in 2D and beyond an open problem.

In this paper, we study the robustness of PEPS under natural perturbations
and show that the LTQO condition, when restricted to specific observables
or regions, allows to prove stability of physical properties for those
observables or regions.  In the context of PEPS, this restricted version
of LTQO has several advantages: On the one hand, it allows to check the
stability of local observables, their derivatives, and correlation
functions under natural perturbations. On the other hand, as it relies
only on the properties of specific operators or regions, it can be
verified numerically by reducing it to an eigenvalue problem. Finally, in
the PEPS framework with its state-centered perspective, there is no need
to additionally check spectral properties of the Hamiltonian, thereby
avoiding this particularly difficult task.  While the motivation for this
work stems from the framework of PEPS, the stability result as such is
independent of PEPS and can be used to assess stability of general quantum
states against a class of physically motivated perturbations.

This paper is structured as follows: In Section II, we introduce the PEPS
formalism and discuss which types of perturbations are natural in the
context of PEPS and parent Hamiltonians.  In Section III, we introduce the
restricted LTQO condition and prove that systems which satisfy LTQO
w.r.t.\ certain observables or regions exhibit robustness against
perturbations. In Section IV, we discuss how the restricted LTQO condition
can be verified for PEPS. We close in Section V by showing that LTQO
itself is stable under perturbations, which in turn allows us to construct
the first examples verifying LTQO without commuting Hamiltonians. In an
appendix, we give a proof that injective Matrix Product States (MPS) satisfy LTQO, implying that they are
stable against general perturbations.

\section{PEPS}

In this section, we introduce the formalism of PEPS and their associated parent Hamiltonians, and define the natural perturbations within this framework.

\subsection{Definition}

We start by recalling the definition and basic properties of PEPS. For the sake of simplicity of the exposition we will concentrate on translationally invariant PEPS $|\Psi\rangle$ on a square lattice. Each PEPS is characterized by a tensor $A\equiv A^{s}_{\alpha,\beta,\gamma,\delta}$ (with a physical index $s=1,\ldots,d$ representing the spin states on a single site, and auxiliary indices $\alpha,\beta,\gamma,\delta=1,\dots,D$), such that $\langle s_1,\ldots,s_N|\Psi\rangle$ is determined by associating the tensor $A^{s_n}$ to each spin $n$, and contracting the auxiliary indices connected by the lattice, as shown in Fig.~\ref{Fig:PEPS3}. For periodic boundary conditions, we also contract the indices on the right boundary with those on the left and the ones pointing up with those down. For open
boundary conditions we can, for instance, set the auxiliary indices at the boundary to a fixed value.

\subsection{Parent Hamiltonian
\label{ssec:parentham}}

\begin{figure}[t]
\includegraphics[width=1\columnwidth]{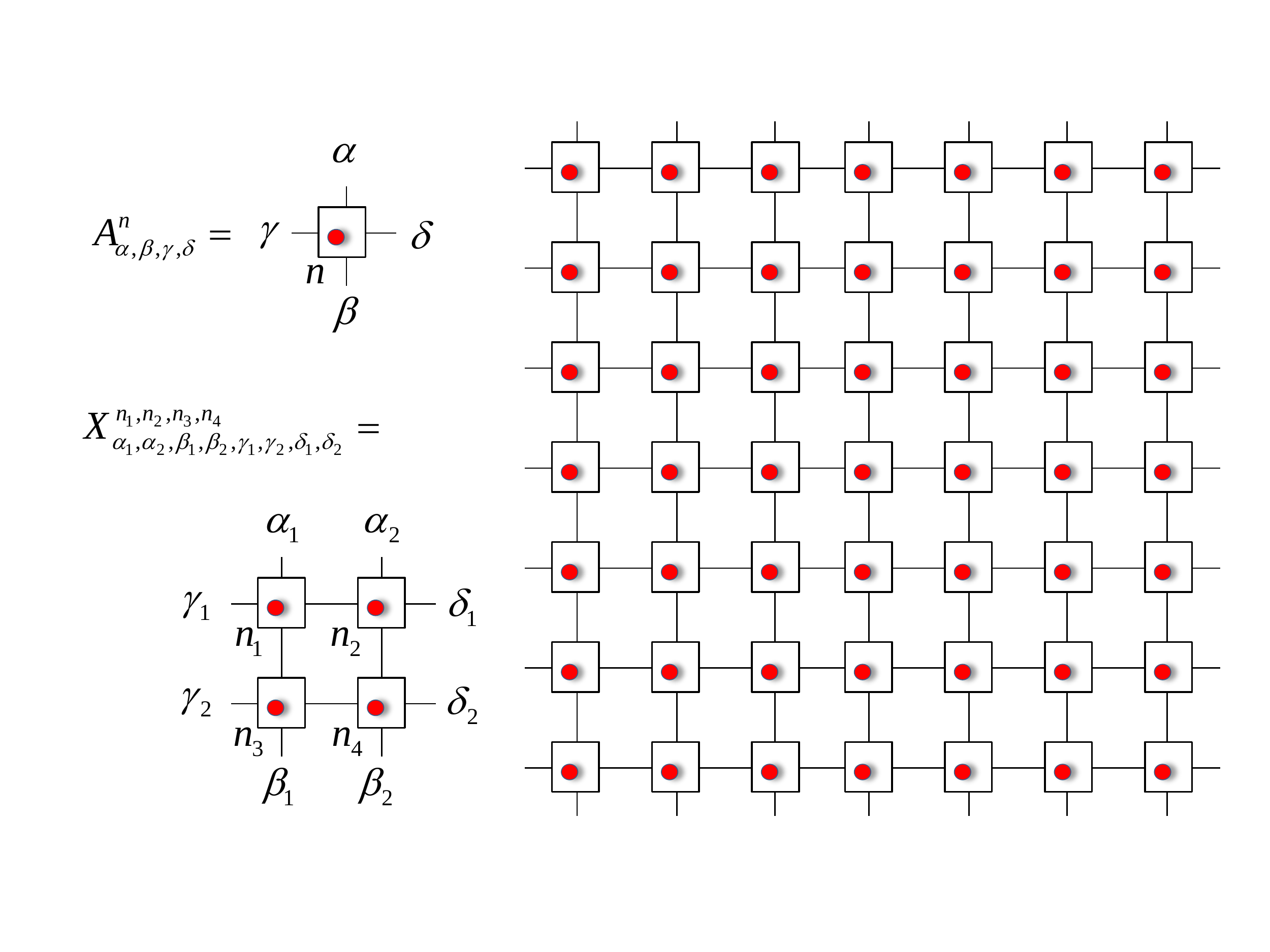}\\
\caption{
Top left: The tensor $A$ has a physical index $s$ (represented by a red point)
and four auxiliary indices, $\alpha,\beta,\gamma,\delta$. The PEPS is
built by contracting all the physical indices along the lattice, as it is
represented by the drawing on the right. At the bottom left, we have the
definition of the tensor $X$ as a contraction of four tensors $A$ along
the auxiliary indices used to define the parent Hamiltonian. 
\label{Fig:PEPS3} 
}
\end{figure}

PEPS are ground states of local frustration free Hamiltonians, which are called parent Hamiltonians. Given a PEPS $|\Psi\rangle$, a parent Hamiltonian has the form: 
\begin{equation*}
 H=\sum_k h_k,
 \end{equation*}
where each $h_k\ge 0$ acts on a finite region $k$ of the lattice. In order to understand the form of $h_k$, we introduce a subspace ${\cal K}_k$ corresponding to the spins in that region. This is the space spanned by all states $|\phi_X\rangle$ which are generated by contracting the tensors $A$ on region $k$, while assigning all possible values to the external auxiliary indices. For instance, if region $k$ is a plaquette composed of $2\times 2$ spins (see Fig. \ref{Fig:PEPS3}), then:
\begin{equation*}
 {\cal K}_k = \mathrm{span}\left\{ |\phi_{\alpha_1\alpha_2\beta_1\beta_2\gamma_1\gamma_2\delta_1\delta_2}\rangle,\;1\le \alpha_i,\beta_i,\gamma_i,\delta_i\le D
 \right\}\,,
 \end{equation*}
where $|\phi_{\alpha_1\alpha_2\beta_1\beta_2\gamma_1\gamma_2\delta_1\delta_2}\rangle=$
 \begin{equation*}
\sum_{n_1,n_2,n_3,n_4} X_{\alpha_1,\alpha_2,\beta_1,\beta_2,\gamma_1,\gamma_2,\delta_1,\delta_2}^{n_1,n_2,n_3,n_4}|n_1n_2n_3n_4\rangle.
 \end{equation*}
For $H$ to be a parent Hamiltonian, ${\cal K}_k$ must coincide with the kernel of $h_{k}$. Furthermore, there must be a way of `growing' the regions $k$ step by step, such that: \emph{(i)}~at each step of joining the spins in neighboring regions, the so-called intersection property is fulfilled;\footnote{The intersection property can be deduced from the injectivity, or more generally, $G$-injectivity \cite{perez-garcia:parent-ham-2d,schuch:peps-sym} property of  a PEPS, but it is less restrictive, as shown for instance in Ref.~\cite{perez-garcia:parent-ham-2d}.} \emph{(ii)} the procedure terminates with a single region containing all spins. The intersection property
simply states that if we join two regions $k_1$ and $k_2$ that intersect in some region $k$ (see
Fig. \ref{Fig:PEPS2}), then:
\begin{equation*}
 \left[{\cal K}_{k_1}\otimes {\cal H}_{k_2 \setminus k}\right]
 \cap
 \left[{\cal H}_{k_1\setminus k}\otimes {\cal K}_{k_2}\right] = {\cal K}_{k_1 
 k_2}\ .
\end{equation*}
Figure \ref{Fig:PEPS2} gives an example of such a construction, where regions are composed of all possible square plaquettes, and  regions $k_1$ and $k_2$ are overlapping plaquettes with two spins in common.

\begin{figure}[t]
\includegraphics[width=1\columnwidth]{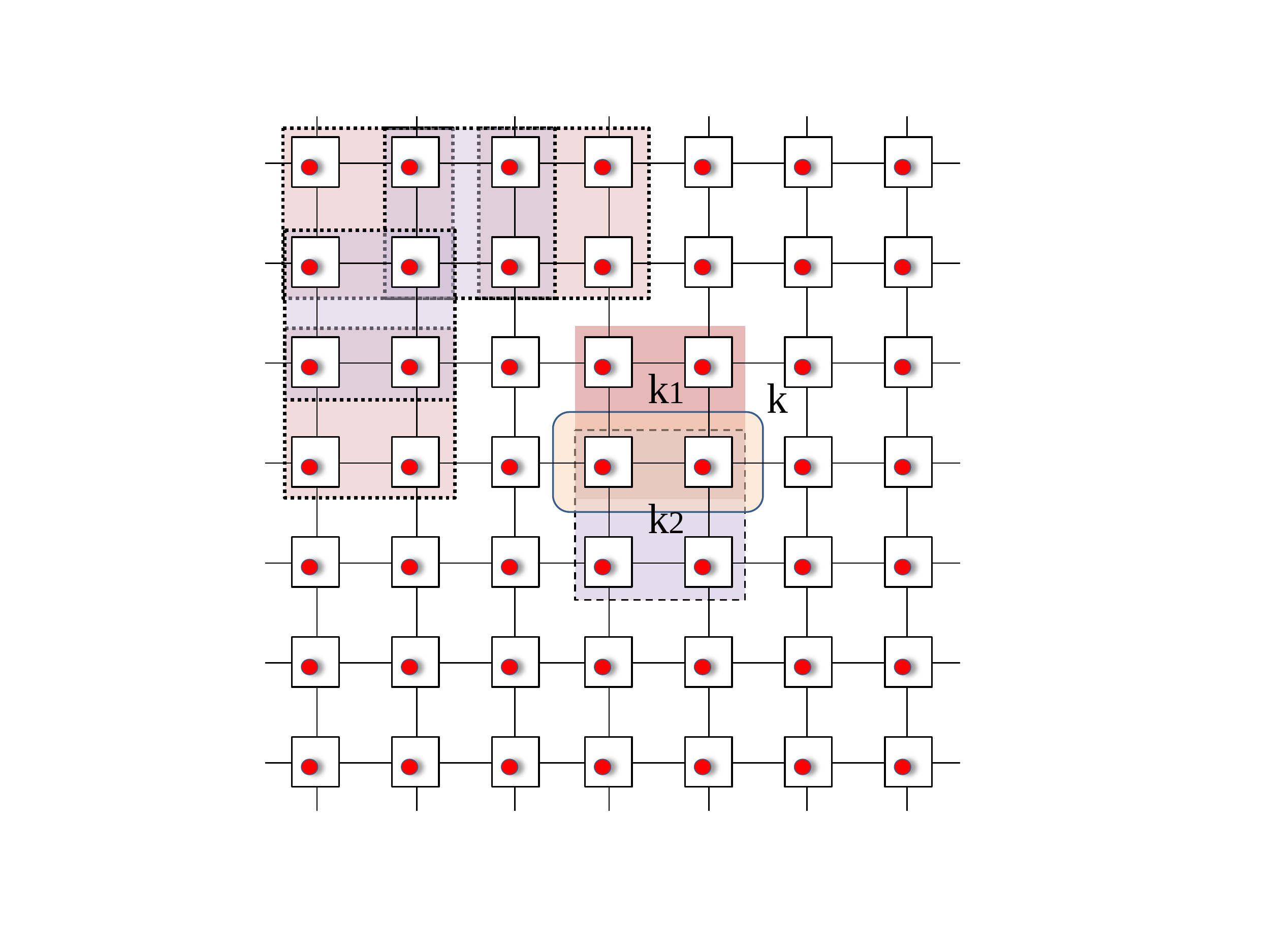}\\
\caption{Construction of parent Hamiltonians for PEPS: The operator
$h_{k_1}$ acts on a region $k_1$ (red region), and its kernel is
spanned by all states that can be obtained by contracting the tensors $A$
on that region with arbitrary boundary conditions on the outgoing
auxiliary indices. We can grow the regions by joining two regions, $k_1$
and $k_2$ (inside the dashed line) that intersect in $k$ (solid line). We
can also block 4 neighboring spin as indicated by the red regions to form
larger spins.}\label{Fig:PEPS2}
\end{figure}

The fact that ${\cal K}_{k}$ coincides with the kernel of $h_k$ ensures that we have a frustration free Hamiltonian, i.e., $h_k|\Psi\rangle=0$. The other properties related to the growth ensure that for any region obtained in the intermediate steps, the groundstate subspace of the part of $H$ acting on that region is spanned by the tensors making up the PEPS if we contract them with arbitrary tensors at the boundary of that region; furthermore, the reduced state of $|\Psi\rangle$ on that region will be supported in that subspace.\footnote{This way, we keep the states $|\psi_b\rangle$, considered below in the definition of LTQO, under control.}

Note that there may be many parent Hamiltonians for a PEPS. On the one hand, any Hamiltonian built up from $\tilde h_k \ge 0$ which are fully supported in ${\cal K}_k^{\perp}$ will be a parent Hamiltonian, and on the other hand, we can choose different regions $k$ (for instance, containing a larger number of spins) to build $H$. The intersection property ensures that no unwanted states appear in the groundstate subspace once we have generated the full Hamiltonian.

\subsection{Perturbations}

The main aim of this paper is the study of the stability of properties of physical systems under perturbations, using the framework of PEPS.  In what follows, we introduce the class of natural perturbations we will study in the context of PEPS.

The most obvious way of perturbing a PEPS is by replacing each tensor in the following way:
 \begin{equation}
 \label{pert-PEPS}
A \to A(\epsilon)=A+\epsilon C\ , \quad \|C\|=1,
\end{equation}
possibly after blocking regions of tensors. Yet, such a perturbation can
lead to a discontinuity in the parent
Hamiltonian~\cite{fernandez:1d-uncle,fernandez-gonzalez:uncle-long}, and
is therefore \emph{unphysical}. Therefore, we will restrict
Eq.~(\ref{pert-PEPS}) to the following class of natural perturbations,
which can be understood as arising from a continuous perturbation of the
parent Hamiltonian.\footnote{We conjecture, based on the intuition built
in \cite{fernandez:1d-uncle,fernandez-gonzalez:uncle-long}, that the
natural perturbations defined in Definition \ref{thm2.3} are exactly those
of the type (\ref{pert-PEPS}) with the added constraint that they lead to
a continuous change in the parent Hamiltonian. Beyond particular classes
of PEPS, like injective ones --where the result is trivial--, or MPS
--where one may use the canonical form \cite{perez-garcia:mps-reps} to
prove it--, we do not have at this point a complete result
connecting~\eqref{pert-PEPS} with \eqref{eq:naturalpert}.}

\begin{defn}\label{thm2.3}
The \emph{natural perturbations} of a PEPS $|\Psi\rangle$ are 
those obtained by applying operators $R(\epsilon)$, with
$\lim_{\epsilon\rightarrow 0}R(\epsilon)=\openone$, 
to fully covering, but non-overlapping, regions of bounded size, i.e.
\begin{equation}
 \label{eq:naturalpert}
 |\Psi(\epsilon)\rangle= R(\epsilon)^{\otimes N'}|\Psi\rangle\ ,
\end{equation}
where $N'$ is the number of regions. 
\end{defn}

To motivate our choice, note that perturbations of the form \eqref{eq:naturalpert} can be understood as arising from a smooth perturbation of the parent Hamiltonian~\cite{schuch:mps-phases,schuch:rvb-kagome}. To see why, start from the frustration-free parent Hamiltonian
$
H=\sum_k h_k
$
for $\ket\Psi$ and let:
\begin{equation}
\label{eq:pert-Ham}
h_k(\epsilon)=((R(\epsilon)^{-1})^{\otimes\kappa})^\dagger 
h_k((R(\epsilon)^{-1})^{\otimes\kappa})\ ,
\end{equation}
where the product $\otimes\kappa$ goes over the sites on which $h_k$ acts.
Note that $R(\epsilon)$ is invertible for small enough $\epsilon$, and
$\lim_{\epsilon \rightarrow 0} h_k(\epsilon)=h_k$.
Then,
\[
H(\epsilon)=\sum_k h_k(\epsilon)
\]
satisfies $h_k(\epsilon)\ge0$ and $h_k(\epsilon)\ket{\Psi(\epsilon)}=0$, i.e., $\ket{\Psi(\epsilon)}$ is a ground state of the frustration-free Hamiltonian $H(\epsilon)$. Indeed, $H(\epsilon)$ is a parent Hamiltonian for $\ket{\Psi(\epsilon)}$: The $h_k(\epsilon)$ have kernels $\mathcal K_k(\epsilon)=R(\epsilon)^{\otimes \kappa}\mathcal K_k$, which -- since $R(\epsilon)$ is invertible -- satisfy the conditions required for parent Hamiltonians discussed in Sec.~\ref{ssec:parentham}.

Note that the construction for $H(\epsilon)$ does not rely on $\ket{\Psi}$ being a PEPS, but only on $H$ being frustration-free, thus our notion of natural perturbations applies to all frustration-free systems.

\section{The LTQO condition}\label{setup}

In this section, we recall the LTQO condition of \cite{michalakis:local-tqo-ffree} and define its restriction to particular observables and regions, which will be the desired checkable property ensuring
stability in the context of PEPS without any spectral assumption. Since
the condition can be introduced and analyzed for general systems (out of
the context of PEPS), we will do so for the sake of generality.

We consider a spin lattice, $X$, in arbitrary spatial dimension, with
corresponding Hilbert space ${\cal H}_X$. We will consider connected
regions of the lattice, $B$, with smooth boundaries, $\partial B$ (see
Fig.~\ref{Fig:Spaces}), and denote by $|B|$ the number of lattice points
in that region and by ${\cal H}_B$ the corresponding Hilbert space for the
spins. We assume a short-range, translationally invariant and frustration-free Hamiltonian, $H_X$, acting on the lattice. We are interested in the properties of the groundstate subspace ${\cal S}_X \in {\cal H}_X$ in the limit
$|X|\to \infty$.

We can write $H_X=H_B+H_{X\setminus B}+H_{\partial B}$, where $H_Y$
includes the terms of $H_X$ acting on $Y$, and denote by ${\cal S}_B\subset
{\cal H}_B$ the groundstate subspace of $H_B$. Note that since $H_X$ is
frustration-free, all states in ${\cal S}_X$ are spanned by vectors in
${\cal S}_B\otimes{\cal S}_{X\setminus B}$.

\begin{figure}[t]
\includegraphics[width=1\columnwidth]{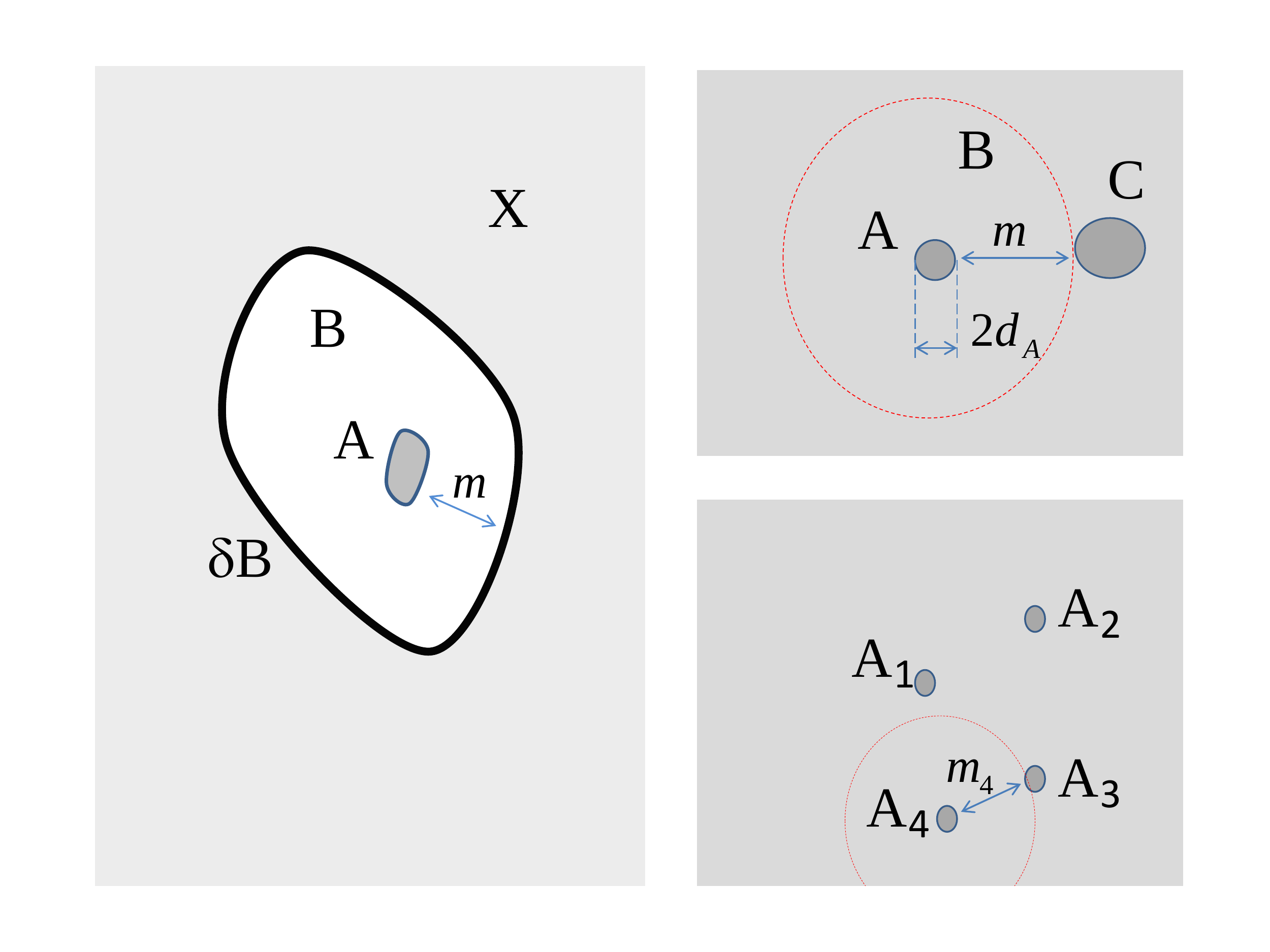}\\
\caption{Left: Setup for the definition of LTQO
(Definition~\ref{Cond22}).
Right: Setup for the decay of two-body correlations,
Proposition~\ref{prop:corr} (top), and many-body correlations,
Proposition~\ref{prop:manycorr} (bottom).
}\label{Fig:Spaces} 
\end{figure}

The LTQO property is related to the sensitivity of local observables to changes in the groundstate far away from the region on which the observable acts. To define it, we divide the lattice $X$ into regions $A\subset B \subset X$, with $A$ and $B$ connected and finite, and denote by $m$ the distance between $A$ and $\partial B$ (see Fig.~\ref{Fig:Spaces}).

\begin{defn}\label{Cond22} We say that a {\em region} $A$ satisfies LTQO, if for all observables $O_a$ supported on $A \subset X$, all regions $B \subset X$ with $A \subset B$ and groundstates $\ket{\Psi_x}\in \mathcal S_X$, and
$\ket{\Psi_b}\in \mathcal S_B$, the following bound holds:
\begin{equation}
 \label{cond22}
 \Big|\frac{\langle \Psi_x | O_a |\Psi_x\rangle} {\langle\Psi_x |\Psi_x\rangle}
 - \frac{\langle \Psi_b|O_a|\Psi_b\rangle}{\langle \Psi_b|\Psi_b\rangle}
\Big| \le \|O_a\| f_A(m)\ ,
 \end{equation}
where $m$ is the distance between $A$ and $\partial B$ (cf.~Fig.~\ref{Fig:Spaces}), and $f_A(m)$ decays superpolynomially in $m$,
i.e.,
\[
\lim_{m\rightarrow\infty} p(m)f_A(m) = 0,
\]
for all polynomials $p(m)$.

We say that a particular {\em observable} satisfies LTQO if it verifies (\ref{cond22}). 

We finally say that a {\em system} satisfies LTQO is all its regions $A$ satisfy it and the function $f$ in (\ref{cond22}) is independent of $A$.
\end{defn}

LTQO for a system was introduced in \cite{michalakis:local-tqo-ffree}. Definition \ref{Cond22} adds its specialization for particular regions and observables, which will be crucial for this paper.

Since we are assuming translational invariance, the exact position of region $A$ on the lattice does not play any role; only its shape matters. Moreover, since LTQO is inherited by subregions, one may restrict to regular shapes, like spheres or cubes. For the purposes of this work, we may think of region $A$ as a single site. 

\begin{prop} The following three are equivalent defintions for LTQO, up to a possible prefactor in $f_A(m)$:
\begin{enumerate}[leftmargin=0cm,itemindent=0.5cm,labelwidth=\itemindent]
\item
For all $\ket{\Psi_b},\ket{\Psi_b'}\in{\cal S}_B$, 
\begin{equation}
 \label{cond22c}
 \Big|\frac{\langle \Psi_b' | O_a |\Psi_b'\rangle} {\langle\Psi_b' |\Psi_b'\rangle}
 - \frac{\langle \Psi_b|O_a|\Psi_b\rangle}{\langle \Psi_b|\Psi_b\rangle}
\Big| \le \|O_a\| f_{A}(m)\ .
\end{equation}
\item For all density operators $\rho_{x}$ and $\rho_b$ supported in ${\cal
S}_{X}$ and $\mc S_B$, respectively,
\begin{equation}
 \label{cond22a}
 \Big|\tr(\rho_x O_a)
 - \tr(\rho_b O_a) \Big| \le \|O_a\| f_{A}(m)\ .
\end{equation}
\item
With $P_B$ the projector onto ${\cal S}_B$ and
$c(O_a)=\frac{\tr(P_BO_a)}{\tr(P_B)}$, 
\begin{equation}
  \label{cond22d}
 \Big\|P_BO_aP_B-c(O_a)P_B\Big\| \le \|O_a\| f_{A}(m)\ .
\end{equation}
\end{enumerate}
Moreover, the following slightly weaker condition is implied by LTQO:
\begin{enumerate}[leftmargin=0cm,itemindent=0.5cm,labelwidth=\itemindent]
\item[4.]
For all operators $Z$ acting on $X\!\setminus\! B$, and all 
$\ket{\Psi_x}\in{\cal S}_X$,
\begin{equation}
 \label{cond22b}
 \Big|\frac{\langle \Psi_x | O_a |\Psi_x\rangle} {\langle\Psi_x |\Psi_x\rangle}
 - \frac{\langle \Psi_x'|O_a|\Psi_x'\rangle}{\langle
\Psi_x'|\Psi_x'\rangle} \Big| \le \|O_a\| f_{A}(m)\ ,
\end{equation}
where we have defined $|\Psi'_x\rangle = Z|\Psi_x\rangle$.
\end{enumerate}
\end{prop}

\begin{proof}
That (\ref{cond22}) implies (\ref{cond22c}) is a simple use of the
triangle inequality and accordingly changing $f$ by $2f$. The reverse
implication follows immediately if we write 
\[
\frac{\tr_{X\setminus B}\big[|\Psi_x\rangle\langle
\Psi_x|\big]}{\langle\Psi_x|\Psi_x\rangle}= \sum_k p_k
|\Psi_b^k\rangle\langle\Psi_b^k|\ .  
\]
where $\ket{\Psi_b^k}\in{\cal S}_b$, $p_k\ge 0$, and $\sum p_k=1$.
That (\ref{cond22a}) implies (\ref{cond22}) is obvious. The converse
follows
directly if we write the spectral decomposition of $\rho_{x}$ and $\rho_b$
and use the convexity of the absolute value. The equivalence between
(\ref{cond22a}) and (\ref{cond22d}) can be seen following the steps of
\cite[Corollary 3]{michalakis:local-tqo-ffree}. Finally, that
(\ref{cond22a}) implies (\ref{cond22b}) can be immediately seen by
defining $\rho_b=R_b/\tr(R_b)$ with $R_b= \tr_{X\setminus B}
(|\Psi_x'\rangle\langle \Psi_x'|)$, so that
\[
 \frac{\langle \Psi_x'|O_a|\Psi_x'\rangle}{\langle \Psi_x'|\Psi_x'\rangle}
=\tr_B(\rho_b O_a)\ , 
\]
and noting that $\rho_b$ is supported in ${\cal S}_B$.
\end{proof}

Note that if all vectors in ${\cal S}_X$ are fully supported on ${\cal
S}_B$, then for all $\ket{\Psi_b}\in{\cal S}_B$ and $\ket{\Psi_x}\in {\cal
S}_X$, there exists a vector $\ket{\Psi_y}\in {\cal S}_{X\setminus B}$ such that
$|\Psi_b\rangle=\langle \Psi_y|\Psi_x\rangle$ and thus (\ref{cond22b})
also implies (\ref{cond22}). This occurs, for instance, if $H_X$ is the
parent Hamiltonian of an injective (or more generally $G$-injective) PEPS
\cite{perez-garcia:parent-ham-2d,schuch:peps-sym}.

Some remarks are in order: \emph{(i)} All conditions have to be fulfilled
independently of the lattice size $|X|$, and therefore also in the
thermodynamic limit. \\
\emph{(ii)}
Eq.~(\ref{cond22}) implies that in the thermodynamic
limit, no two states in ${\cal S}_X$ can be distinguished locally by means of $O_a$. \\
\emph{(iii)}
Eq.~(\ref{cond22c}) implies that $O_a$ cannot distinguish different states in
${\cal S}_B$, as long as the boundary of $B$ is far enough from the region
where we measure. \\
\emph{(iv)} We will, in the following, generally assume that $B$ is
spherical. Indeed, (\ref{cond22}) cannot be modified to depend on $|B|$
(and thus the shape of $B$) in a non-trivial way: On the one hand, an
exponential dependence on $|B|$ would override the scaling of $f(m)$ and
invalidate the condition. On the other hand, a polynomial scaling $p(|B|)$
can be removed by choosing a spherical $B'\subset B$ with identical $m$,
and observing that in (\ref{cond22a}), $|\tr(\rho_x O_a)-\tr(\rho_b
O_a)|=|\tr(\rho_x O_a)-\tr(\rho_{b'} O_a)| \le \|O_a\|p(|B'|)f_{A}(m)$,
where $p(|B'|)$ is
polynomial in $m$ and can thus be absorbed in $f$.

\section{Implications of the LTQO condition}

We will now analyze which restrictions the LTQO condition imposes on a
system. We start by showing a superpolynomial decay of correlations and
then use this to give the desired stability result.

\subsection{Correlation functions}

We show here that if an observable $O_a$ satisfies LTQO, then correlation
functions must decay superpolynomially\footnote{The speed of the decay in
the correlations is the same that appears in the definition of LTQO.} with
the distance.

\begin{prop}
\label{prop:corr}
If $O_a$ satisfies LTQO, then for any observable $O_b$ acting on \mbox{$X\!\setminus\!\! B$}
(cf.~Fig.~\ref{Fig:Spaces}), 
\begin{equation}
 \label{corr}
 \big|\langle O_a O_b \rangle-\langle O_a\rangle \langle O_b\rangle \big| \le 
\|O_a\| \, \|O_b\|\, f_{A}(m),
 \end{equation}
where the expectation value is taken in any normalized state 
$\ket{\Psi_x}\in \mc S_X$.
\end{prop}

\begin{proof} We can always write $O_b=P_b-Q_b$, where both $P_b,Q_b\ge 0$
and $\|O_b\|=\max\{\|P_b\|,\|Q_b\|\}$, so that we just have to prove (\ref{corr})
for $P_b\ge 0$. Defining $|\Psi_x'\rangle=\sqrt{P_b}|\Psi_x\rangle$, we
have
 \begin{equation}
 \langle O_a P_b \rangle =
 \frac{\langle \Psi_x'|O_a|\Psi_x'\rangle}
   {\langle \Psi_x'|\Psi_x'\rangle}
 \langle \Psi_x|P_b|\Psi_x\rangle\ .
 \end{equation}
Using condition (\ref{cond22b}) and the fact that the last factor is bounded by $\|P_b\|$, we obtain (\ref{corr}) (up to a factor of $2$).
\end{proof}

We can iterate Eq.~(\ref{corr}) to prove that also many-site correlation
functions decay fast. To this end, let us consider some regions $A_1, A_2,
\ldots, A_M$, and denote by $m_k$ the shortest distance between $A_k$
and the rest of the regions (see Fig.~\ref{Fig:Spaces}). Then,
\begin{prop}
\label{prop:manycorr}
For any set of observables $O_{a_k}$ verifying LTQO and acting on regions A$_k$,
 \begin{equation}
 \label{corr2}
 \Big|\langle \prod_{k=1}^M O_{a_k} \rangle- \prod_{k=1}^M\langle O_{a_k}\rangle\Big| \le \prod_{k=1}^M\|O_{a_k}\|\sum_{k=1}^Mf_{A_k }(m_k).
 \end{equation}
\end{prop}

\begin{proof}
We define
\[
 r_n=\langle \prod_{k=1}^n O_{a_k} \rangle- \langle\prod_{k=1}^{n-1}
O_{a_k}\rangle \langle O_{a_n}\rangle,\quad s_n=\prod_{k+1}^M\langle
O_{a_k}\rangle\ , 
\]
with $s_M=1$ and $r_1=0$.  We have that
\[
 \Big|\langle \prod_{k=1}^M O_{a_k} \rangle- \prod_{k=1}^M\langle
O_{a_k}\rangle\Big|= \Big| \sum_{n=1}^M r_n s_n \big| \le \sum_{n=1}^M
|r_n| |s_n|\ .
\]
Then, Eq.~(\ref{corr2}) follows from
\[
 |s_n|\le \prod_{k+1}^M \|O_{a_k}\|, \quad |r_n|\le f_{A_n}(m_n)
\prod_{k=1}^n \|O_{a_k}\|\ ,
\]
where in the last inequality we have used (\ref{corr}), and the fact that
$\|AB\|\le\|A\|\,\|B\|$.

\end{proof}

\subsection{Robustness against perturbations}

We are interested in seeing how the properties of the ground state
subspace ${\cal S}_X$ change if we modify the states locally. In
particular, we want to know the behavior of the expectation values of
local observables:
\begin{equation}
 \label{form1}
 o_a(\epsilon):=
 \frac{\langle \Psi_x(\epsilon)| O_a|\Psi_x(\epsilon)\rangle}
 {\langle \Psi_x(\epsilon) |\Psi_x(\epsilon)\rangle}
 \end{equation}
if we perturb every state $\ket{\Psi_x}\in{\cal S}_X$ as follows:
 \begin{equation}
 \label{Psieps}
 |\Psi_x(\epsilon)\rangle=R_X(\epsilon)|\Psi_x\rangle
 \end{equation}
in the limit $|X|\rightarrow \infty$, where
$R_X(\epsilon)=R(\epsilon)^{\otimes |X|}$.  Here,
$R(\epsilon)=\Id+\epsilon Z$, where $Z$ is an operator acting on a single
lattice site with $\|Z\|=1$, and $\epsilon$ sufficiently small.
Note that we do not need to restrict ourselves to $Z$ acting on a single
lattice site; in fact, we can always group spins into bigger spins and
assume $Z$ acts on a single spin of the new lattice. Recall that these
were exactly the natural perturbations in the context of PEPS. 

For simplicity in the notation we will restrict to the case of
translational invariant perturbations, but the results hold true with the
same proofs in the case of a site-dependent perturbation of the form
$\otimes_{i\in X} R_i(\epsilon) |\Psi\rangle$.\footnote{There is another
motivation for this type of perturbation that goes beyond PEPS. Imagine
each spin is weakly coupled to a local environment. After the action of
the noise for time $\epsilon$, the system is in a convex combination of
states of the form $\otimes_{i\in X} R_i(\epsilon) |\Psi\rangle$, where
each $R_i(\epsilon)$ is $\epsilon$-close to the identity. Hence, LTQO
implies that the system is also stable against this type of dissipative
noise.}

We will now show that LTQO implies robustness for local observables: If the observable $O_a$ satisfies LTQO, $o_a(\epsilon)$ is continuous at $\epsilon=0$, and
its first derivative at that point is finite. Moreover, if not only the observable $O_a$ but also the single site region satisfies LTQO, then all higher order derivatives are also finite at $\epsilon=0$.

\begin{prop}
If the observable $O_a$ satisfies LTQO, then $o_a(\epsilon)$ is continuous at $\epsilon=0$. 
More specifically, there exists a function $k_{A}(\epsilon)\rightarrow
0$ (as $\epsilon\rightarrow 0$) which is independent of $|X|$, such that
$$|o_a(\epsilon)-o_a(0)|\le \|O_a\|k_{A}(\epsilon)$$ for all
lattice sizes larger than some $|X_\epsilon|$.  
\end{prop}

\begin{proof}
The idea of the proof is to decompose the perturbation $R_X(\epsilon)$,
Eq.~(\ref{Psieps}), into two parts: One part is far away from $O_a$ and
can thus be dealt with using LTQO, while the other part can be bounded
directly as it only acts on a restricted region. 

 We start by choosing a number $m := m(\epsilon)$ and a region $B_{m}\supset
A$, such that the distance between $A$ and $\partial B_m$ is $m$. Define
$|\Psi_x'(\epsilon)\rangle = R_{X\setminus
B_m}(\epsilon)|\Psi_x(0)\rangle$, i.e., $|\Psi_x(0)\rangle$ is only
modified outside of region $B_m$. From now on we will omit the dependence
of all the states and operators on $\epsilon$ to facilitate reading. We
write the numerator of
(\ref{form1}) as
 \begin{equation}
 \langle \Psi_x | O_a |\Psi_x\rangle =
 \langle \Psi_x' | O_a |\Psi_x'\rangle +
 \langle \Psi_x' | T_a |\Psi_x'\rangle\ ,
 \end{equation}
where $T_a=R_{B_m}^\dagger O_a R_{B_m}-O_a$. In the same way, we replace the denominator by
 \begin{equation}
 \langle \Psi_x | \Psi_x\rangle =
 \langle \Psi_x' | \Psi_x'\rangle +
 \langle \Psi_x' | S_a |\Psi_x'\rangle\ ,
 \end{equation}
where $S_a=R_{B_m}^\dagger R_{B_m}-\Id$. With simple manipulations, we write
 \begin{equation}
 \label{form2}
 |o_a(\epsilon)-o_a(0)| \le \big|  \frac{\langle \Psi'_x | O_a |\Psi'_x\rangle}{\langle \Psi'_x | \Psi'_x\rangle} -o_a(0)
 \big| + h(\epsilon)\ ,
 \end{equation}
where
 \begin{equation}
 h(\epsilon)= \frac{\|O_a\|\|S_a\|+\|T_a\|}{1-\|S_a\|}\ .
 \end{equation}
In order to bound this term, we write $T_a =(R_{B_m}-\Id)^\dagger O_a
(R_{B_m}-\Id) + (R_{B_m}-\Id)^\dagger O_a + O_a (R_{B_m}-\Id)$, so that
 \begin{equation}
 \|T_a\| \le \|O_a\| \|R_{B_m}-\Id\| \left(2+\|R_{B_m}-\Id\|\right)\ .
 \end{equation}
The same bound applies to $\|S_a\|$ when replacing $\|O_a\|$ by $1$. Using
the binomial expansion of $R_{B_m}$, we have
 \begin{equation}
 \|R_{B_m}-\Id\|\le(1+ |\epsilon|)^{(|A|+2m)^2}-1\ ,
 \end{equation}
where we have used that $|B_m|\le (|A|+2m)^2$. Choosing
$m(\epsilon)=|\epsilon|^{-1/2+x}$ with $x\in(0,1/2)$, we
have $\|R_{B_m}-\Id\|\to 0$ in the limit $\epsilon\to 0$, and thus
$h(\epsilon)\to 0$. Finally, using (\ref{cond22b}) we can bound the first
term of (\ref{form2}) by $\|O_a\| f_{A}(m)$, which vanishes in that
limit as well.
\end{proof}

\begin{prop}
If an observable $O_a$ satisfies LTQO, then $do_a(\epsilon)/d\epsilon$ is finite at
$\epsilon=0$. Formally, the limit $$\lim_{|X|\rightarrow \infty}
\left.\frac{do_a(\epsilon)}{d\epsilon}\right|_{\epsilon=0}$$ exists and is
finite.  
\end{prop}

\begin{proof}
In order to determine $o'_a=do_a(\epsilon)/d\epsilon$ at $\epsilon=0$, we
first have to take the derivative of $\ket{\Psi(\epsilon)}$ as given in
(\ref{Psieps}). We split $o'_a$ into two parts: \emph{(i)} one
part corresponding to the  derivative involving lattice sites included in
$A$; \emph{(ii)}~the rest. The first is obviously finite. The second can
be written as $\tilde o'_a$
where
\begin{equation}
 \tilde o_a(\epsilon)=
 \frac{\langle \Psi_x| \tilde R(\epsilon) O_a|\Psi_x\rangle}
 {\langle \Psi_x |\tilde R(\epsilon)|\Psi_x\rangle}
 \end{equation}
with $\tilde R(\epsilon) = (\Id + \epsilon W)^{\otimes |X\setminus A|}$ and
$W=Z+Z^\dagger+ \epsilon Z^\dagger Z$. Taking the derivative and setting
$\epsilon=0$ we obtain
 \begin{equation}
 \tilde o'_a = \sum_{n\notin A} \left[\langle W_n O_a\rangle - \langle W_n\rangle\langle O_a\rangle\right],
 \end{equation}
where the sum is extended to all sites not belonging to $A$, $W_n$ denotes
$W$ acting on site $n$, and the expectation values are taken in the
(normalized) state $\ket{\Psi_x}$. Using that the correlation functions
decay faster than any polynomial, Eq.~(\ref{corr}), the sum converges in
the limit $|X|\rightarrow\infty$.  
\end{proof}

One can extend the proof to any higher order derivative.

\begin{prop}\label{prop8}
If both an observable $O_a$ and the single site region satisfy LTQO, then
\begin{equation}
  \lim_{|X|\rightarrow\infty} \frac{d^no_a( \epsilon)}{d\epsilon^n}
 \Big|_{\epsilon=0}
 \end{equation}
exists and is finite.
\end{prop}

The proof is analogous to the one above, although a bit more involved. It
relies again on the fact that connected correlation functions decay
sufficiently fast, Eq.~(\ref{corr2}).

Note that to prove the finiteness of the derivatives, we only use the
decay of the correlation functions. The full power of the LTQO condition is
only used directly in the continuity proof. This is a formal proof in this
context that with exponential decay of correlations, one can only expect
first order or infinite-order phase transitions. Having LTQO rules out the
first-order ones. 

If we do not have LTQO for the single site region and we have it only for a particular observable $O_a$, we cannot guarantee Proposition \ref{prop8} to hold. In this case, we can only deduce continuity and bounded first derivative. This rules out first and second order phase transitions witnessed by
$O_a$, but leaves open the possibility of higher-order ones. Finally, note that only the weakest condition (\ref{cond22b}) has been used for the proofs in this section.  

\section{LTQO in PEPS
\label{sec:LTQO-and-peps}}

We have seen that LTQO ensures stability for PEPS under a class of natural
perturbations. In this Section, we will analyze how to detect LTQO in PEPS,
and discuss PEPS-specific implications of the stability condition.

\subsection{Detecting LTQO in PEPS}

Consider a translationally invariant PEPS $\ket{\Psi_x}$ 
(see upper-left part of Fig.~\ref{Fig:PEPS}) with some boundary condition
(which by the very definition of LTQO will play no role). The PEPS is
fully characterized by a tensor $A$ with some physical index
$n=1,\ldots,d$ and auxiliary indices $\alpha_k=1,\ldots,D$, where $d$ is
the dimension of the spin and $D$ the \emph{bond dimension}. In order to
investigate the LTQO property for this state, we consider an observable
$O_1$, with $\|O_1\|=1$, acting on the central spin in the figure, which we will call spin 1;
note that we can always block spins such that the operator $O_1$ only acts
on a single effective spin. We now define a one-dimensional structure of
tensors by layer-wise blocking tensors around spin 1: The first tensor
corresponds to spin 1 itself. The second is obtained by contracting all
tensors around spin 1 (marked green in Fig.~\ref{Fig:PEPS}). The third one
contains those next to the previous layer (marked violet in
Fig.~\ref{Fig:PEPS}), and so on. The resulting chain of
tensors is represented in the lower-left part of Fig.~\ref{Fig:PEPS}. We
denote them by $B[1],B[2],\ldots,B[m]$, where the dimensions of the
physical and the auxiliary indices now grow with the layer $m$. That is,
in this representation the PEPS $\ket{\Psi_x}$ has the form of a (non
translationally invariant) matrix product state (MPS).

\begin{figure}[t]
\includegraphics[width=1\columnwidth]{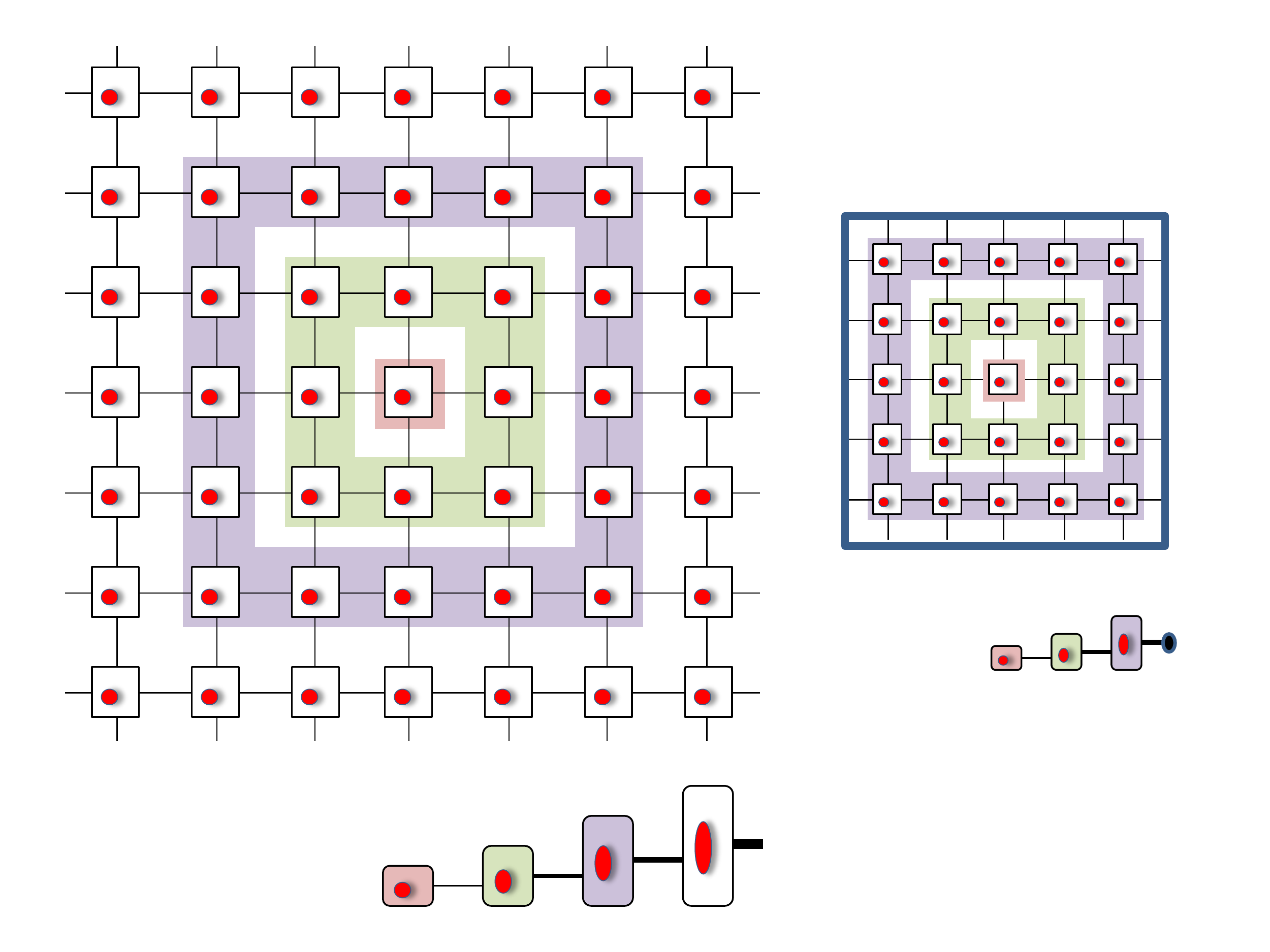}
\caption{\label{Fig:PEPS}
Verifying LTQO in PEPS. Left: By consecutively blocking regions around
the central spin, we can map the PEPS onto a one-dimensional Matrix
Product State (MPS). Right: The effect of the boundary condition (dark
blue) on the central spin can be mapped to an eigenvalue problem for the
transfer operators of the one-dimensional chain (see text).}
\end{figure}

Let us now consider a region $B$ in the original lattice centered around
spin 1, containing layers $1$ to $m$. Any state $\ket{\Psi_b}\in {\cal
S}_B$ can be obtained by contracting those layers with an arbitrary tensor
on its boundary (see upper-right part of Fig.~\ref{Fig:PEPS}). In terms of
the MPS representation (lower part), this just corresponds to contracting
a vector with the auxiliary index on the right. Thus,
\[
 |\Psi_R\rangle_B = \sum_{n_1,\ldots,n_m} (B[1]^{n_1}| 
\cdots B[m]^{n_m} |R) |n_1,\ldots,n_m\rangle
\]
where $B[k]^n$ are $D_{k-1}\times D_{k}$ matrices, and $|B[1]^n)$ and
$|R)$ are vectors of dimensions $D_1$ and $D_m$, respectively (we have
used curly brackets to denote vectors acting on the auxiliary indices, in
order to avoid confusion with the physical spin degrees of freedom). As it
is standard in MPS theory, in order to determine expectation values of
local observables acting on spin 1, it is convenient to define the
following completely positive maps:
\begin{eqnarray*}
 {\cal E}_1(X)&=&\sum_{i_1,j_1} |B^{j_1}) \langle j_1|X|i_1\rangle
(B^{i_1}|\ ,\\
 {\cal E}_n(X)&=&\sum_{i_n} B^{i_n\dagger} X B^{i_n}
\end{eqnarray*}
for $n=2,\ldots,m$, as well as,
\[
 {\cal M}_m={\cal E}_m \circ \ldots {\cal E}_2 \circ{\cal E}_1\ .
\]
We can thus write:
 \begin{equation}
\label{eq:LTQO-top-eq}
 \frac{\langle \Psi_R|O_1|\Psi_R\rangle }
 {\langle \Psi_R|\Psi_R\rangle} =
 \frac{(R|{\cal M}_m(O_1)|R) }
 {(R|{\cal M}_m(\Id)|R)}\ .
 \end{equation}
We will have LTQO for $O_1$ whenever this quantity becomes independent of the vector
$|R)$ in the limit $m\to \infty$ via a rapidly decaying function $f(m)$. 

In order to numerically verify the presence of LTQO using
Eq.~(\ref{eq:LTQO-top-eq}), one finds the maximum and minimum generalized
eigenvalues $\lambda(O_1)$ of the eigenvalue equation
 \begin{equation*}
 {\cal M}_m (O_1) |R) = \lambda(O_1) {\cal M}_m(\Id) |R)\ ,
 \end{equation*}
which can be done using Lanczos methods, together with approximate
contraction of the quasi-1D tensor network.
Defining
 \begin{equation*}
 \epsilon_m = \lambda_{\rm max}(O_1)-\lambda_{\rm min}(O_1),
 \end{equation*}
we then have
 \begin{equation*}
 \Big| \frac{\langle \Psi_R|O_1|\Psi_R\rangle }
 {\langle \Psi_R|\Psi_R\rangle} -
 \frac{\langle \Psi_S|O_1|\Psi_S\rangle }
 {\langle \Psi_S|\Psi_S\rangle}\Big|  \le \epsilon_m
 \end{equation*}
and thus, it only remains to check that $\epsilon_m$ decays sufficiently
fast with $m$.\footnote{Note that $\epsilon_m$ depends on the chosen
observable. If we want to check LTQO for the one-site region, we have to
maximize $\epsilon_m$ among all possible observables $O_1$ with
$\|O_1\|=1$.}

\subsection{Implications of LTQO for PEPS simulations}

If a PEPS possesses LTQO for a certain local observable $O_1$, this implies that in order to compute expectation values of $O_1$, we can choose any boundary condition
$|R)$ we like;  in particular, we can choose $|R)$ to be a product state.
If the boundary is at a distance $m$ from the observable, this implies
that the boundary as seen by the observable is a Matrix Product Operator
(MPO) with bond dimension $D^{2m}$ (obtained by blocking the tensors in
Fig.~\ref{Fig:PEPS} in radial slices). In particular, if
$f(m)=O(e^{-\alpha m})$, then the bond dimension required to compute the
value of $\langle O_1\rangle$ in the thermodynamic limit up to precision
$\epsilon$ scales polynomially in $1/\epsilon$. Thus, LTQO provides a
formal justification of the approximate PEPS contraction scheme in which
the boundary is approximated by an MPO at every
step~\cite{verstraete:2D-dmrg}.

\section{Stability of LTQO}

In this Section, we prove that in the presence of a spectral gap, the LTQO condition for a system can only
disappear when closing the (local) gap. This is important since it allows
us to infer LTQO for a whole neighborhood
of systems rather than only for isolated points in Hamiltonian space.

Given a system $H_X=\sum_{x\in X} h_x$, we say that it has local gap if
there exists a constant $\gamma>0$ such that for all $|X|$ and all
spherical regions $R\subset X$, the Hamiltonian $H_R=\sum_{x\in R} h_x$
has a spectral gap at least $\gamma$ above the ground state energy.

\begin{thm}
\label{thm:ltqo-stable}
Consider a Hamiltonian $H_X=\sum_x h_x$ which is short ranged and
frustration free\footnote{Though we assume it for simplicity, the
hypothesis of frustration freeness can be relaxed. The notion of LTQO for
frustrated Hamiltonians with a low-energy subspace of energy-splitting
$\epsilon$ separated by a uniform gap $\gamma \gg \epsilon$ from the rest
of the spectrum, is well-defined if one assumes the existence of
subregions with low-energy subspaces separated by local-gaps to the rest
of the spectrum. In that case, LTQO is defined with respect to the
projections onto the local, low-energy subspaces with energy below
$\epsilon$, instead of the exact groundstate subspace of frustration-free
Hamiltonians.} (but not necessary translationally invariant), and let
$\epsilon_x\ge 0$ for all $x\in X$. Assume that (i) the system $H_X$ has
LTQO with some superpolynomially decaying function $\hat f(m)$; (ii)~for all
$0\le \delta_x\le \epsilon_x$, the perturbed Hamiltonian
\begin{equation}
\label{eq:app:def-ham}
H^{\vec{\delta}}_X=\sum_x h_x+ \delta_x k_x
\end{equation}
has uniform (in $\delta$) local gap and $k_x$ acts on the same sites
as $h_x$, where we assume $\|h_x\|, \|k_x\|\le 1$ for all $x$. Then, the
perturbed system $H^{\vec{\epsilon}}_X$ has LTQO. 
\end{thm}

It is crucial for the proof that we assume LTQO for the whole system, and not just for a particular region or observable.

To prove the result, we will use the following result from
Ref.~\cite{bachmann:quasi-adiabatic}. 

\begin{lem}(Theorem 3.4 in \cite{bachmann:quasi-adiabatic}) 
\label{lemma:bachmann}
Let $Y$ be any region of a system $X$ and $Y_R$ the region enlarged by
sites at distance $\le R$ from $Y$.  Consider a smooth path of
Hamiltonians on $X$, $H(s)=H_0+\Phi(s)$, $0\le s\le 1$, with uniformly
bounded local terms, bounded derivatives, and a uniform lower bound on the
spectral gap, and for which $\Phi(s)$ is supported on $Y$.  Let $P_0$ and
$P_1$ be the projector onto the ground space of $H(0)$ and $H(1)$,
respectively. Then, there exists a unitary operation $V_R$ acting on $Y_R$
such that (in operator norm) 
\[
\|P_0-V_RP_1V_R^\dagger\| \le \tilde f(R)\ ,
\]
where $\tilde f(m)$ decays superpolynomially.
\end{lem}
\noindent 
Note that it will be decisive that the Lemma makes no
assumption about the rank of $P_0$ and $P_1$.

\noindent\emph{Proof of Theorem~\ref{thm:ltqo-stable}.}
In order to prove LTQO for the deformed system (\ref{eq:app:def-ham}), we
consider a spherical region $A$ and subsequently add concentric rings $B,
C, D$, such that their boundaries are separated by $\frac{m}{3}$. We
denote the union of the regions $A,B,C,D$ by $Y$, and the projector onto
the ground space of the original Hamiltonian $H_X$ in region $Y$ by $P$.

Let $P^{AB}$ denote the projection onto the ground space of Hamiltonian
\[
H_{AB}^{\vec\delta} = \sum_{x\in X} h_x + \sum_{x\in{A,B}} \delta_x k_x\ ,
\]
i.e., where the perturbation only acts in regions $A$ and $B$.  
Since we assume a local gap in the theorem, Lemma~\ref{lemma:bachmann}
implies the existence of a unitary $V_{ABC}$ (supported on regions
$A,B,C$) such that
\begin{equation}
\label{eq:app:lemma10-abc}
\|P^{AB}-V_{ABC}PV_{ABC}^\dagger\|\le \tilde{f}(\frac{m}{3})\ .
\end{equation}
Using successive triangle inequalities, the submultiplicativity and
unitary invariance of the operator norm, and
Eq.~(\ref{eq:app:lemma10-abc}), we find (with 
$c=\frac{\tr{PV_{ABC}^\dagger O_A V_{ABC}}}{\tr{P}}$)
\begin{widetext}
\[
\|P^{AB}O_AP^{AB}-cP^{AB}\|\le
\|V_{ABC}PV_{ABC}^\dagger O_A
V_{ABC}PV_{ABC}^\dagger-cV_{ABC}PV_{ABC}^\dagger\|+
3\|O_A\|\tilde f(\frac{m}{3})\
.
\]
The left part can be further bounded using the LTQO condition for
$H_X$ which says that $\|PO_{ABC}P-c'P\|\le \hat
f(\frac{m}{3})\|O_{ABC}\|$ for all $O_{ABC}$ supported in the union of
regions $A,B,C$ (in particular for $O_{ABC}=V_{ABC}^\dagger O_A V_{ABC}$),
where $c'=\tfrac{\tr\,{PO_{ABC}}}{{\tr\,{P}}}\equiv c$, which yields
\[
\|P^{AB}O_AP^{AB}-cP^{AB}\|\le
\|O_A\|\hat f(\frac{m}{3})
+3\|O_A\|\tilde f(\frac{m}{3})
\le 4\|O_A\|f(\frac{m}{3})\ ,
\]
where $f$ is a superpolynomially decaying upper bound to $\tilde f$ and
$\hat f$.

Another application of Lemma~\ref{lemma:bachmann} proves the existence of
a unitary $V_{BCD}$ such that
\[
\|P^{ABCD}-V_{BCD}P^{AB}V_{BCD}^\dagger\|\le f(\frac{m}{3})\ ,
\]
where $P^{ABCD}$ is the projector onto the ground space of
$H^{\vec\delta}_X$.
Again,
\[
\|P^{ABCD}O_AP^{ABCD}-cP^{ABCD}\|\le \|V_{BCD}P^{AB}V_{BCD}^\dagger
O_A V_{BCD}P^{AB}V_{BCD}^\dagger- c V_{BCD}P^{AB}V_{BCD}^\dagger\|+
3\|O_A\|f(\frac{m}{3})\ .
\]
Since in the first term on the r.h.s., $O_A$ commutes with $V_{BCD}$,
$V_{BCD}^\dagger V_{BCD}$ cancels, and we find (using unitary invariance of
the norm)
\[
\|P^{ABCD}O_AP^{ABCD}-cP^{ABCD}\|\le \| P^{AB}O_AP^{AB}-cP^{AB}\|+
3\|O_A\|f(\frac{m}{3})\le 7\|O_A\|f(\frac{m}{3})\ ,
\]
\end{widetext}
with $c$ defined as above. Using the characterization given in Corollary 3
of Ref.~\cite{michalakis:local-tqo-ffree}, one can see that the actual
value of the constant $c$ plays no role in the definition of LTQO, and
therefore, we have shown LTQO for the perturbed system.
\hspace*{\fill}$\square$

This result can be used to construct  new examples of
systems verifying LTQO. For instance, it is shown in
\cite[Appendix E]{schuch:mps-phases} that if we start with a system with
LTQO and made out of commuting terms (such as the toric code or quantum
double models), small perturbations of the type (\ref{eq:pert-Ham}) verify
the hypothesis of the theorem. In this way, we can give the first 2D
examples of systems with non-commuting Hamiltonians satisfying LTQO.

\section{Conclusions}

In this paper, we have analyzed the stability of a PEPS under physical
perturbations to the local tensor which defines it. We have shown how
restricting the LTQO condition \cite{michalakis:local-tqo-ffree} to
particular observables and regions gives a checkable criterion which makes
this assignment between the PEPS and the local tensor robust.  This
robustness translates then to any situation in which this assignment is
exploited, with examples ranging from classifying quantum phases in
locally interacting spin systems \cite{schuch:mps-phases, schuch:peps-sym}
to approximating numerically ground states of 2D local
Hamiltonians~\cite{jordan:iPEPS}.

\subsection*{Acknowledgments}

We are indebted to Frank Verstraete for many useful comments on this work.
JIC acknowledges support by the EU project AQUTE, the DFG SFB 631 and
Exzellenzcluster NIM, and Catalunya Caixa. SM acknowledges funding
provided by the Institute for Quantum Information and Matter, an NSF
Physics Frontiers Center with support of the Gordon and Betty Moore
Foundation through grant GBMF1250 and by the AFOSR grant FA8750-12-2-0308.
DPG acknowledges support from Spanish grants MTM2011-26912 and QUITEMAD,
and European CHIST-ERA project CQC. NS acknowledges support by the
Alexander von Humboldt foundation.

\appendix

\section{LTQO for injective MPS}

In this appendix, we give a formal proof of the following theorem:

\begin{thm}\label{thm:MPS-LTQO}
Parent Hamiltonians of translationally invariant, injective MPS satisfy LTQO.
\end{thm} 

Among MPS experts the above result has been known for some time, but we think that a rigorous proof would illuminate some of the key aspects of LTQO as it relates to the concept of {\it insensitivity of the bulk to boundary conditions}. Before proving the result, we recall the requisite basic machinery from the MPS literature. In particular, we note that a translationally invariant MPS is equivalent to a 1D PEPS. It is, hence, given by a collection of $D\times D$ matrices $(A_i)_{i=1}^d$, with $d$ the local physical dimension. Since contraction in this case reduces to matrix multiplication, for each chain with $N$ spins, the MPS reads:
\begin{equation*}
|\psi\rangle = \sum_{i_1,\ldots,i_N=1}^d \tr(A_{i_1}\cdots A_{i_N})|i_1\cdots i_N\rangle
\end{equation*}

An MPS is called injective if there exists a length $R \ge 1$, such that
the map $$K_R(X): X\mapsto \sum_{i_1,\ldots,i_R=1}^d \tr(XA_{i_1}\cdots
A_{i_R})|i_1\cdots i_R\rangle$$ is injective. The minimal such $R$ is
called the injectivity length. By the quantum Wielandt inequality of
\cite{sanz:wielandt}, 
the injectivity length is known to be upper-bounded by $(D^2-d+1)D^2$.
Hence by blocking at most $(D^2-d+1)D^2$ spins we can assume without loss
of generality that $R=1$. Injective MPS are the unique ground states of
their parent Hamiltonians~\cite{fannes:FCS,perez-garcia:mps-reps}, which have a uniform gap above the ground state \cite{fannes:FCS}. Moreover, parent Hamiltonians of injective MPS also verify the {\it local-gap condition} of \cite{michalakis:local-tqo-ffree}. That is, for any region of $L$ consecutive spins, the Hamiltonian $H_L=\sum_{i=1}^{L-1}h_{i,i+1}$, whose groundstate subspace is $$\ker(H_L)=\{K_L(X)| X\in M_{D\times D}\},$$
has a uniform (in $L$) spectral gap~\cite{fannes:FCS}. This allows us to conclude from Theorem \ref{thm:MPS-LTQO} and the main result in \cite{michalakis:local-tqo-ffree} that:

\begin{cor}
Parent Hamiltonians of translationally invariant, injective MPS have a stable spectral gap against arbitrary quasi-local perturbations.
\end{cor}
Note that the above corollary combined with the quasi-adiabatic
continuation
technique~\cite{hastings:quasi-adiabatic,bachmann:quasi-adiabatic} implies
stability of the groundstate subspace with respect to properties of local
observables.

To show Theorem \ref{thm:MPS-LTQO}, we will rely on the canonical form of
MPS stated in \cite{perez-garcia:mps-reps}. Any injective MPS can be
represented by a set of $D \times D$ matrices $\{A_i\}_{i=1}^d$, such that
the completely-positive and trace-preserving map $\E$ given by
$\E(X)=\sum_{i=1}^d A_i XA_i^\dagger$, has a non-degenerate eigenvalue of
modulus $1$ corresponding to $\Lambda$, where $\Lambda$ is a diagonal,
positive, full-rank matrix with $\tr(\Lambda)=1$. If we denote the second
largest (in magnitude) eigenvalue of $\E$ as $\lambda_2$, then it follows
that the map $\E$ has a spectral gap given by $1-|\lambda_2|$. 

\begin{proof} [Proof of Theorem \ref{thm:MPS-LTQO}]
We consider a region $B$ with spins $1,\ldots, 2m+l$ and region $A$ with spins $m+1,\ldots, m+l$ as well as an unnormalized ground state of $H_B$ given by $X$:
$$|\psi_X\rangle=\sum_{i_1,\ldots,i_{2m+l}=1}^d \tr(XA_{i_1}\cdots A_{i_{2m+l}})|i_1\cdots i_{2m+l}\rangle\; .$$
To show LTQO it is enough to prove that for any observable $O_A$ acting on region $A$:
$$\left| \frac{\langle \psi_X|O_A|\psi_X\rangle}{\langle \psi_X|\psi_X\rangle}-\tr(\E_{O_A}(\Lambda))\right| \le \|O_A\| f(m)$$
with $f(m)$ exponentially decaying in $m$ and
\begin{widetext}
$$\E_{O_A}(X)=\sum_{i_{m+1},\ldots,i_{m+l},j_{m+1},\ldots,j_{m+l}=1}^d \langle  j_{m+1}\cdots j_{m+l}|O_A|i_{m+1}\cdots i_{m+l}\rangle A_{i_{m+1}}\cdots A_{i_{m+l}}XA_{j_{m+l}}^\dagger \cdots A_{j_{m+1}}^\dagger\; .$$
Set $g(O_A)=\left| \frac{\langle \psi_X|O_A|\psi_X\rangle}{\tr(XX^\dagger \Lambda)}-\tr(\E_{O_A}(\Lambda))\right|$. It is shown in \cite[Lemma 5.2.(2)]{fannes:FCS} that $g(O_A)\le \|O_A\|f(m)$ with $f$ exponentially decaying with $m$. Then, 
\begin{align*}\left| \frac{\langle \psi_X|O_A|\psi_X\rangle}{\langle \psi_X|\psi_X\rangle}-\tr(\E_{O_A}(\Lambda))\right|&\le
\left| \frac{\langle \psi_X|O_A|\psi_X\rangle}{\langle \psi_X|\psi_X\rangle}-\frac{\langle \psi_X|O_A|\psi_X\rangle}{\tr(XX^\dagger \Lambda)}\right|
+ \left| \frac{\langle \psi_X|O_A|\psi_X\rangle}{\tr(XX^\dagger \Lambda)}-\tr(\E_{O_A}(\Lambda))\right|\\
&\le \|O_A\|g(\1)+g(O_A)\le 2\|O_A\|f(m),
\end{align*}
as desired.  
\end{widetext}
\end{proof}

\end{document}